\documentclass{sig-alternate-05-2015}

\usepackage{algorithmic}
\usepackage[algoruled]{algorithm2e}

\newtheorem{theorem}{Theorem} 

\newtheorem{lemma}[theorem]{Lemma}

\newtheorem{theorem*}{Theorem}

\newtheorem{observation}[theorem]{Observation}

\usepackage{tikz}
\usetikzlibrary{arrows,chains}
\usepackage{multirow}
\usepackage{multicol}

\newcommand{\junk}[1]{}

\begin{document}
\setcopyright{acmcopyright}

\title{Asymptotically Optimal Approximation Algorithms for Coflow Scheduling}
\numberofauthors{3}
\author{
\alignauthor
Hamidreza Jahanjou\thanks{All authors are with the College of Computer and Information Science, Northeastern University, Boston, MA 02115, USA. Email: $\{$hamid, erezk, rraj$\}$@ccs.neu.edu.  Supported in part by NSF grants Grant CSR-1217981 and CCF-1422715, and a Google Research Award.}
\alignauthor
Erez Kantor$^*$\\
\alignauthor
Rajmohan Rajaraman$^*$\\
}

\maketitle

\begin{abstract}
Many modern datacenter applications involve large-scale computations composed of multiple data flows that need to be completed over a shared set of distributed resources.  Such a computation completes when all of its flows complete.  A useful abstraction for modeling such scenarios is a {\em coflow}, which is a collection of flows (e.g., tasks, packets, data transmissions) that all share the same performance goal.  

In this paper, we present the first approximation algorithms for scheduling coflows over general network topologies with the objective of minimizing total weighted completion time.  We consider two different models for coflows based on the nature of individual flows: circuits, and packets.  We design constant-factor polynomial-time approximation algorithms for scheduling packet-based coflows with or without given flow paths, and circuit-based coflows with given flow paths.  Furthermore, we give an $O(\log n/\log \log n)$-approximation polynomial time algorithm for scheduling circuit-based coflows where flow paths are not given (here $n$ is the number of network edges).

We obtain our results by developing a general framework for coflow schedules, based on interval-indexed linear programs, which may extend to other coflow models and objective functions and may also yield improved approximation bounds for specific network scenarios.  We also present an experimental evaluation of our approach for circuit-based coflows that show a performance improvement of at least \%22 on average over competing heuristics.
\end{abstract}

\section{Introduction}\label{sec:intro}
In modern datacenter network applications, a large-scale computation over a big data set is often composed of multiple tasks or multiple data flows that need to be completed over a shared set of distributed resources.  Such a computation completes when all the data flows in the computation complete.  A useful abstraction for modeling such scenarios is that of a {\em coflow} \cite{ChowdhuryStoica2012}, which is a collection of flows (e.g., tasks, packets, data transmissions) that all share the same performance goal.  Coflows arise frequently in distributed computing and datacenters \cite{DBLP:conf/nsdi/Al-FaresRRHV10,ChowdhuryStoica2012,DBLP:conf/sigcomm/ChowdhuryZMJS11,DBLP:conf/sigcomm/HongCG12,DBLP:conf/spaa/QiuSZ15,DBLP:conf/sigcomm/ZatsDMBK12,rapier}. A prominent example is parallel data processing; in certain MapReduce computations, the reduce phase at a particular reducer can begin only after all the relevant data from the map phase has arrived at the reducer \cite{DBLP:journals/cacm/DeanG08}.  Similar examples occur in Dryad \cite{DBLP:conf/eurosys/IsardBYBF07}, and Spark \cite{DBLP:conf/hotcloud/ZahariaCFSS10}. 

In this paper, we study the problem of designing coflow schedules that minimize total weighted completion time.  We consider two different models for coflows based on the nature of individual flows: circuits and packets.

\begin{itemize}
\item {\bf Circuit-based coflows} a flow is a connection request (data transmission) from source to sink on a network.
\item {\bf Packet-based coflows} a flow is a packet that needs to be routed from source to sink on a network.
\end{itemize}

Before going into formal definitions, Figure \ref{fig:coflow_example} illustrates a circuit-based coflow scheduling scenario. Three coflows $A$, $B$, and $C$ require bandwidth assignment and scheduling on a triangle network with unit edge capacities. Coflow $A$ consists of two flows $A_1$ of size $2$ and $A_2$ of size $1$ whereas the other two have one flow each of size $1$. Three solutions are proposed. In (s1), each flow is given a bandwidth of $1/2$; all flows are scheduled to run in parallel. The total completion time is $4+2+4=10$. In (s2), coflow $A$ has the highest priority followed by coflow $B$ and $C$. The total completion time is $2+2+4=8$. But, it is possible to do even better by observing that flow $C$ can run at the same time as either flow $A_2$ or flow $B$; this gives rise to an optimal solution, shown in (s3), with a total completion time of $4+2+1=7$.

We now present the models and problem formulations, our main results and techniques, and a review of related work.

\begin{figure*}
\centering
\begin{tikzpicture}[scale=1]

\node[draw,circle,thick] (x) at (0,0) {x};
\node[draw,circle,thick] (y) at (1,1.73) {y};
\node[draw,circle,thick] (z) at (2,0) {z};

\draw[thick] (x)--(y) node[midway,right] {$1$};
\draw[thick] (y)--(z) node[midway,left] {$1$};
\draw[thick] (x)--(z) node[midway,above] {$1$};

\draw[->,green, very thick] (-0.5, 0) -- (0.5,1.73) arc (180:0:0.5) node[above,midway]{$\sigma(C)=2$} -- (2.5,0) node[right,midway]{$C$};

\draw[->,blue, very thick] (-1, 0) node[below]{$\sigma(A_1)=2$} -- (0,1.73) node[left,midway]{$A_1$};
\draw[->,blue, very thick] (0, -0.5) -- (2,-0.5) node[right]{$A_2$} node[below,midway]{$\sigma(A_2)=1$};

\draw[->, red, very thick] (0,-1) node[left]{$B$} -- (2,-1) node[midway, below]{$\sigma(B)=1$};
\node at (1,-2) {(N)};
\draw[fill, blue] (4,-0.3) rectangle (6,-0.1); \node[left] at (4,-0.2) {$b(A_1)$}; \node[right] at (6,-0.2) {$1/2$}; 
\draw[fill, blue] (4,0.2) rectangle (5,0.4); \node[left] at (4,0.3) {$b(A_2)$}; \node[right] at (5,0.3) {$1/2$}; 
\draw[fill, red] (4,0.7) rectangle (5,0.9); \node[left] at (4,0.8) {$b(B)$}; \node[right] at (5,0.8) {$1/2$}; 
\draw[fill, green] (4,1.2) rectangle (6,1.4); \node[left] at (4,1.3) {$b(C)$}; \node[right] at (6,1.3) {$1/2$}; 
\node at (5,-2) {(s1)};

\draw[>=latex,->] (4,-0.5) -- (6,-0.5) node[below,right] {time} node[midway, below]{$0\;\;\;\, 1\;\;\;\, 2\;\;\;\, 3\;\;\;\, 4$}; 
\draw (4,-0.5) -- (4,2); 
\draw[fill, blue] (8,-0.4) rectangle (9,0); \node[left] at (8,-0.2) {$b(A_1)$}; \node[right] at (9,-0.2) {$1$}; 
\draw[fill, blue] (8,0.1) rectangle (8.5,0.5); \node[left] at (8,0.3) {$b(A_2)$}; \node[right] at (8.5,0.3) {$1$}; 
\draw[fill, red] (8.5,0.6) rectangle (9,1); \node[left] at (8,0.8) {$b(B)$}; \node[right] at (9,0.8) {$1$}; 
\draw[fill, green] (9,1.1) rectangle (10,1.5); \node[left] at (8,1.3) {$b(C)$}; \node[right] at (10,1.3) {$1$}; 
\node at (9,-2) {(s2)};

\draw[>=latex,->] (8,-0.5) -- (10,-0.5) node[below,right] {time} node[midway, below]{$0\;\;\;\, 1\;\;\;\, 2\;\;\;\, 3\;\;\;\, 4$}; 
\draw (8,-0.5) -- (8,2); 
\draw[fill, blue] (13,-0.4) rectangle (14,0); \node[left] at (12,-0.2) {$b(A_1)$}; \node[right] at (14,-0.2) {$1$}; 
\draw[fill, blue] (12.5,0.1) rectangle (13,0.5); \node[left] at (12,0.3) {$b(A_2)$}; \node[right] at (13,0.3) {$1$}; 
\draw[fill, red] (12,0.6) rectangle (12.5,1); \node[left] at (12,0.8) {$b(B)$}; \node[right] at (12.5,0.8) {$1$}; 
\draw[fill, green] (12,1.1) rectangle (13,1.5); \node[left] at (12,1.3) {$b(C)$}; \node[right] at (13,1.3) {$1$}; 
\node at (13,-2) {(s3)};

\draw[>=latex,->] (12,-0.5) -- (14,-0.5) node[below,right] {time} node[midway, below]{$0\;\;\;\, 1\;\;\;\, 2\;\;\;\, 3\;\;\;\, 4$}; 
\draw (12,-0.5) -- (12,2); 
\end{tikzpicture}
\caption{Three possible bandwidth assignments for the set of coflows $\{A,B,C\}$. Coflow $A$ has two flows ($A_1$ of size $2$ and $A_2$ of size $1$); the other two have each one flow of size $1$. The objective is to minimize sum of completion times. (N) The network is a triangle where each edge has unit capacity. Flows are drawn around the edges along with their size. (s1) A bandwidth of $1/2$ is assigned to each flow; total completion time is $10$. (s2) Priority is given to coflow $A$ then to coflow $B$ and then to coflow $C$; total completion time is $8$. (s3) An optimal solution; total completion time is $7$.}\label{fig:coflow_example}
\end{figure*}
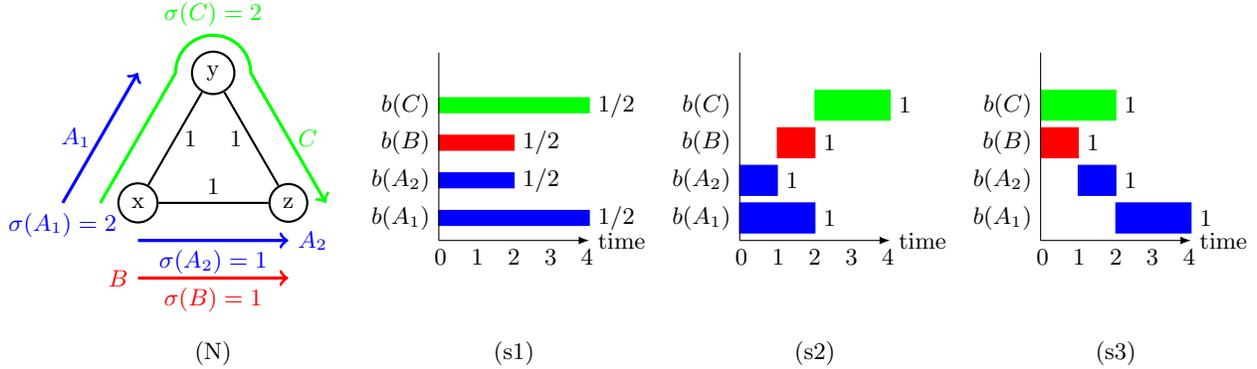

\subsection{Models and Notation}\label{sec:model}
Regardless of the model, the following notation is used. The set of all flows is denoted by $F=\{f_1, ..., f_n\}$. We define a \emph{coflow} $F_i$ to be a subset $\{f^i_1, ..., f^i_{n_i}\}$ of flows such that the set of all coflows, denoted by $\mathcal{F}$, provides a partition of $F$. Each flow $f_i$ has a source $s_i$, a destination $d_i$, and, depending on the model, a size $\sigma_i$. Additionally, each flow has a release time $r_i$ at which point it becomes available. Notice that in our formulation each flow has a release time whereas in previous works, release times are at the level of coflows.

We also have a network modeled as a directed graph $G=(V,E)$ which, in the case of circuit-based coflows, also features edge capacities $\{c(e)\}_{e\in E}$.

The completion time of a flow $f_i$ is denoted by $c_i$. The completion time $C_k$ of a coflow $F_k$ is defined to be the completion time of its last flow; that is,
\begin{equation*}
C_k = \max_{f\in F_k} c_{f}.
\end{equation*}

The goal is to minimize the weighted sum of coflow completion times 
\begin{equation}
\label{eq:obj} 
\overline{C}=\sum_k \omega_k C_k = \sum_k \omega_k \max_{f\in F_k} c_{f}.
\end{equation} 

Note that if there is only one coflow, the problem reduces to makespan minimization. Similarly, if every coflow has one flow, the problem reduces to minimizing total weighted (flow) completion time. Therefore, in a sense, coflow scheduling bridges two well-studied minimization objectives. In particular, this means that hardness results for either of these problems also hold for coflow scheduling. Below, we consider the particulars of each model.


\subsubsection*{Circuits}
In this model, a flow $f_i$ is a connection request (data transmission) from a source node $s_i$ to a destination node $d_i$ on a network $G=(V,E)$ with edge capacities $\{c(e)\}_{e\in E}$. A coflow is a set of connection requests. To each connection request in $F$, we need to assign a path $p_i$ (if not given already) and a bandwidth $b_i$ such that $\overline{C}$ is minimized without violating edge capacities. In general, the bandwidth of a connection request is an integrable function of time, $b_i(t)$, and its completion time $c_i$ is defined as the smallest value satisfying
\[
\int_{0}^{c_i} b_i(t)dt = \sigma_i,
\]
where $\sigma_i$ is the size of the connection request.

As we will see, one can assume, without loss of generality, that the bandwidth functions are piece-wise constant. We can further subdivide this problem into two cases: (a) paths are given for each connection request and we only need to assign bandwidths, (b) paths are not given; hence, each connection request requires routing and bandwidth assignment.

\subsubsection*{Packets}
In this case, each flow $f_i$ is a packet residing (at its release time $r_i$) on its source node $s_i$ waiting to be routed to its destination node $d_i \in V - \{s_i\}$. A coflow is a set of packets. Each link $e\in E$ can serve at most one packet at a time (i.e. edge capacities are set to 1). A coflow $F_k$ completes at time $C_k$ when all of its packets have reached their destination. As before, the goal is to minimize $\overline{C}$.

This model can be seen as a discrete version of circuit-based coflows. As in the previous model, we will consider two subproblems: (a) a path is given for each packet; only scheduling is needed
(b) paths are not given; both routing and scheduling are required for each packet.

\subsection{Our Results and Techniques}\label{sec:results}
We give asymptotically optimal  approximation algorithms for minimizing total
weighted completion time under the different models defined
above.

\begin{itemize}
\item 
We present $O(1)$-approximation algorithms for circuit-based coflows with given paths, packet-based coflows with given paths, and packet-based coflows where paths are not given.
\item
We present an $O(\log n/\log\log n)$-approximation algorithm for circuit-based coflows where paths are not given.
\item
We also present an experimental evaluation of our approach for circuit-based coflows which shows \%22 or more performance improvement on average over competing heuristics.
\end{itemize}
Note that, all of the above problems are at least NP-hard. Table \ref{tbl:results} summarizes our results.

\junk{
We give constant-ratio approximation algorithms for all cases except circuit-based coflows where paths are not given. 

As far as we know, circuit-based scheduling has not been investigated. We will show that circuit-based coflow scheduling with fixed paths is $\mathbf{NP}$-hard via a reduction from\\ $1|r_j,pmtn|\sum \omega_j C_j$. The variable-path case is harder. We show that the hardness of circuit-based coflow scheduling follows from hardness congestion minimization on directed graphs.

All of our results are asymptotically optimal. 
}

\begin{table}
\begin{center}
\begin{tabular}{cccc}
Model & Paths & Approx. & Hardness\\
\hline
\multicolumn{1}{|c}{\multirow{2}{*} {Packet-based}} & \multicolumn{1}{|c|}{given} & \multicolumn{1}{c|}{\multirow{3}{*}{$O(1)$}} & \multicolumn{1}{c|}{\multirow{2}{*}{$\mathbf{APX}$-hard}} \\ \cline{2-2}
\multicolumn{1}{|c}{}
 & \multicolumn{1}{|c|}{not given} & \multicolumn{1}{c|}{} & \multicolumn{1}{c|}{}\\ \cline{1-2}\cline{4-4}
\multicolumn{1}{|c}{\multirow{2}{*} {Circuit-based}} & \multicolumn{1}{|c|}{given} &\multicolumn{1}{c|}{} &\multicolumn{1}{c|}{$\mathbf{NP}$-hard} \\ \cline{2-4}
\multicolumn{1}{|c}{}
 & \multicolumn{1}{|c|}{not given} & \multicolumn{1}{c|}{$O(\frac{\log |E|}{\log\log |E|})$} & \multicolumn{1}{c|}{$\Omega(\frac{\log |E|}{\log\log |E|})$}\\
\hline
\end{tabular}
\caption{Approximation ratio of our algorithms and corresponding lower bounds for different models.}
\label{tbl:results}
\end{center}
\end{table}

A noteworthy contribution of our work is a common framework for developing coflow schedules, which applies to all the models we study. The framework consists of three parts. First, we reformulate the coflow scheduling problem as an instance of minimizing total weighted {\em flow} completion time problem with precedence constraints in form of depth-1 in-trees. Second, we devise an interval-index linear program for each problem instance, and then reduce the problem of finding flow schedules that minimize total weighted completion time (with aforementioned precedences) to multiple instances of the problem of finding flow schedules that minimize makespan. Third, we perform rounding.  The specifics of each step varies across models. In particular, the last step is achieved by a careful rounding of the associated linear program solution.  Though our approximation bounds are optimal to within constant factors, we have not attempted to optimize these constants.  We believe that our framework provides a general approach for solving coflow scheduling problems under diverse models and objectives. 


The algorithms given for circuit-based coflows provide the first provable approximation bounds, and apply in the general setting of arbitrary capacitated networks, arbitrary source-sink pairs, arbitrary demands, and arbitrary release times.  

To the best of our knowledge, packet-based coflows have not been studied earlier, though related models have been analyzed, as we discuss in previous work. 

\subsection{Previous Work}\label{sec:prev_work}
\subsubsection*{Coflows}
The notion of coflow was first proposed by Chowdhury and Stoica as {\em ``a networking abstraction to express the communication requirements of prevalent data parallel programming paradigms"} \cite{ChowdhuryStoica2012, DBLP:conf/sigcomm/ChowdhuryZMJS11}. In \cite{ChowdhuryZhongStoica2014}, Chowdhury et al. present effective heuristics for coflow scheduling without release times on a non-blocking switch network. They also show that coflow scheduling is $\mathbf{NP}$-hard for $P\times P$ switches, $P\geq 2$, via a reduction from the concurrent open shop scheduling problem. Approximation algorithms for coflow scheduling on a non-blocking switch with (coflow) release times are given in~\cite{DBLP:conf/spaa/QiuSZ15,KM-coflow-SPAA16,shafiee2017brief,Ahmadi2017}. Zhao et al. consider coflow scheduling over general network topologies; thereby, routing the flows becomes an additional requirement \cite{rapier}. They provide heuristics based on the observation that simultaneous routing and scheduling of coflows is required to obtain good performance.


\subsubsection*{Machine scheduling and task-based coflows}
\junk{
Machine scheduling problems are described \cite{Graham1979287} by three values $\alpha | \beta | \gamma$ where the first field represents the machine environment, the second field describes job constraints and the third field specifies the objective function to be minimized. The set of possible machine environments relevant to us is $\{ 1,P,R,J, O \}$ where $1$ denotes single machine, $P$ identical parallel machines, $R$ unrelated parallel machines, $J$ job shop, and $O$ open shop. Possible job constraints are $pmtn$ for preemption, $p_j=1$ for unit processing times and $r_j$ for release times. The objective function is either makespan $C_{\max}$ or weighted sum of completion times $\sum \omega_j C_j$. For more information, we refer the reader to one of the many surveys in this area (e.g. \cite{Hall:1996:AAS:241938.241939}).
}

Although coflow scheduling is a relatively recent development, the related problem of scheduling jobs consisting of tasks on parallel machines where a job completes when all of its tasks complete, has been studied before. This model is known in the literature as \emph{order scheduling}; the goal is to minimize the (weighted) sum of job completion times. In their survey, Leung et al., classify order scheduling problems into three categories based on capabilities of the machines \cite{Leung2005}. Specifically, in the fully-dedicated case, each machine can process only one type of component (task). This model, sometimes denoted by $G||\sum \omega_j C_j$, is also known as the concurrent open shop scheduling since, in contrast to the latter, components of a job can be processed at the same time \cite{DBLP:conf/fsttcs/GargKP07}. In the fully-flexible case, all machines can process all types of component (task). Lastly, in the arbitrary case, unrestricted subsets of machines can process arbitrary subsets of components (tasks). 

Machine scheduling problems are described \cite{Graham1979287} by three values $\alpha | \beta | \gamma$ where $\alpha$ represents the machine environment, $\beta$ describes job constraints and $\gamma$ specifies the objective function to be minimized. 

The task-based coflow scheduling model, where a flow is simply a task, corresponds to the fully-flexible case on (unrelated) parallel machines. For this model, Blocher and Chhajed observe that the problem is strongly $\mathbf{NP}$-complete even when all machines are identical \cite{NAV:NAV3}. Correra, Skutella and Verschae give a $13.5$ approximation algorithm for this problem on unrelated machines \cite{Correa2009}.


\subsubsection*{Circuit-based coflows}
Circuit-based coflow scheduling with fixed paths is a generalization of $1|r_j,pmtn|\sum \omega_j C_j$ which is strongly $\mathbf{NP}$-hard \cite{Labetoulle1984245}. On the other hand, the version with variable paths is related to congestion minimization. Chuzhoy et al. show that congestion minimization on directed graphs is hard to approximate within  $\Omega(\frac{\log |E|}{\log\log |E|})$ unless $\mathbf{NP}$ has $N^{O(\log\log N)}$-time randomized algorithms~\cite{DBLP:conf/stoc/ChuzhoyGKT07}.  

\subsubsection*{Packet-based coflows}
The problem of scheduling packets on store-and-forward networks with given paths is equivalent to the job shop scheduling problem with unit processing times. Leighton, Maggs and Rao \cite{Leighton94packetrouting} show that there exists a schedule achieving a makespan of $O(C+D)$ where $C$ and $D$ denote maximum edge congestion and maximum path dilation respectively. In \cite{DBLP:journals/combinatorica/LeightonMR99}, Leighton et al.\ give a polynomial-time algorithm for finding such a schedule. Peis et al. \cite{DBLP:conf/waoa/PeisSW09} show that for all $\epsilon > 0$ there is no $(6/5-\epsilon)$-approximation algorithm for the packet scheduling problem with given paths (minimizing makespan) unless $\mathbf{P}=\mathbf{NP}$. 
An $O(1)$-approximation algorithm for $J|r_{ij}, p_{ij}=1|\sum_S C_S$ is given in \cite{Queyranne02approximationalgorithms}; the authors consider a generalized minsum objective function.

On the other hand, when paths are not given, Srinvasan and Teo give the first constant-ratio approximation algorithm for packet routing and scheduling to minimize makespan \cite{TeoSrinv}. Koch et al. \cite{kochmessage} extend this result to messages which consist of packets having the same source and sink. 

The rest of this paper is organized as follows.  We study circuit-based coflows in \S\ref{sec:circuit_based} and packet-based coflows in \S\ref{sec:packet_based}. We present experimental results in~\S\ref{sec:experiments} and conclude in~\S\ref{sec:open_problems}.

\section{Circuit-based coflows}\label{sec:circuit_based}
We now consider our main model. In this setting, we have a network $G=(V,E)$ with edge capacities $\{c(e)\}_{e\in E}$. A coflow $F_i$ has a weight $\omega_i \in \mathbb{R}^{\geq 0}$ and consists of a number of connection requests (data transmissions) $f^i_j$. Each connection request has a source $s^i_j \in V$, a sink $d^i_j \in V$, a release time $r^i_j \in \mathbb{R}^{\geq 0}$ and a size $\sigma^i_j \in \mathbb{R}^{\geq 0}$. A solution must specify a path $p^i_j$ from source to sink (if not given), and a bandwidth  function $b^i_j$ for each connection request $f^i_j$. In general, the bandwidth $b^i_j(t)$ is an integrable function of time such that
\begin{align}\label{eq:continuous_completion}
\int_{r^i_j}^{c^i_j} b^i_j(t)dt &= \sigma^i_j,
\end{align}
where $c^i_j$, the completion time of $f^i_j$, is the smallest value for which the equality holds.
In addition, we require that edge capacities are respected. Hence, any solution must satisfy 
\begin{equation}\label{eq:flow_req}
\forall e \in E\ : \sum_{f^i_j\in P(e)} b^i_j(t) \leq c(e)
\end{equation}
for all $t$, where $P(e)$ is the set of paths that use edge $e$.

There are two natural versions of this problem. In one version, a path $p^i_j$ is given as a part of the input for each connection request $f^i_j$ and we need to assign a bandwidth $b^i_j$ to each request with the goal of minimizing total weighted (coflow) completion time $\overline{C}$ without violating edge capacities. Note that any network topology in which there is a unique path between pairs of vertices, e.g. trees or non-blocking switches, falls into this category. In another version, paths are not given; hence, assigning a (single) path to each connection request becomes an additional requirement.

The rest of this section is organized as follows. First, we present two key lemmas used in deriving our results. Next, in \S\ref{sec:circuit_based_paths_given}, we give a $17.6$-approximation algorithm for the circuit-based coflow problem with given paths. Finally, in \S\ref{sec:circuit_based_paths_not_given}, an $O(\frac{\log |E|}{\log\log |E|})$-approximation algorithm is presented for the case where paths are not given. 
\subsubsection*{Two lemmas}
The following two lemmas are simple but crucial. The first one (Lemma \ref{lem:const_bandw}) states that, without loss of generality, we can assume that bandwidth functions are piece-wise constant.
\begin{lemma}\label{lem:const_bandw}
Suppose that there exist a capacitated network $(G,c)$ and a set of feasible flows $\{f_1, ..., f_n\}$, where each flow $f_i$ is specified by a triple $(p_i,b_i(t), r_i)$ consisting of its path, bandwidth, and release time. Given two time points $T_1<T_2$ such that $r_i \leq T_1$ for all $i \in[n]$, there exists a set of feasible constant bandwidths $\{b'_i\}_{1\leq i\leq n}$ delivering the same amount as $\{b_i(t)\}_{1\leq i \leq n}$ for every flow over the time interval $[T_1,T_2]$.
\end{lemma}
\begin{proof}
Let $\sigma_i$ denote the amount of flow delivered on path $p_i$; that is, $\sigma_i=\int_{T_1}^{T_2} b_i(t)dt$. Define new bandwidths $b'_i=\sigma_i/(T_2-T_1)$. Clearly, the volume  delivered remains the same and release times are respected. Moreover, for any edge $e$, the sum of flow bandwidths passing through $e$ satisfies $\sum_{i\in P(e)} b_i' \leq c(e)$, where $P(e)$ is the set of paths that use $e$. Indeed, assume that there exists an edge $e'$ for which $\sum_{i\in P(e')} b_i' > c(e')$. Then, it follows that
\begin{align*} 
\sum_{i\in P(e')} b_i' &= \sum_{i\in P(e')} \sigma_i / (T_2-T_1) > c(e') \\
&\Rightarrow \sum_{i\in P(e')} \frac{\int_{T_1}^{T_2} b_i(t)dt}{T_2-T_1} > c(e')\\
&\Rightarrow \sum_{i\in P(e')} \int_{T_1}^{T_2} b_i(t)dt > (T_2-T_1)c(e')
\end{align*}
which is a contradiction since the total volume delivered can not be more than $(T_2-T_1)c(e')$.
\end{proof}

The second lemma (Lemma \ref{lem:disc_bandw}) shows that it is possible to \emph{discretize} the bandwidth usage among all flows over a given path.
\begin{lemma}\label{lem:disc_bandw}
Suppose that there exist a path $p$ in a capacitated network $(G,c)$ and a set of feasible flows $\{f_1, ..., f_n\}$ over this path, where each flow $f_i$ has a bandwidth $b_i(t)$ and a release time $r_i$. Given two time points $T_1<T_2$ such that $r_i \leq T_1$ for all $i \in[n]$, there exists a set of feasible bandwidths $\{b''_i\}_{1\leq i\leq n}$ such that at any point in time $t\in [T_1,T_2]$, there is at most one active flow (i.e. with non-zero bandwidth). Furthermore, for every flow, the delivery amount over the time interval $[T_1,T_2]$ is as before.
\end{lemma}
\begin{proof}
Let $c_m = \min_{e\in p} c(e)$ denote the minimum edge capacity of the path. Also, let $\sigma_i$ denote the volume corresponding to $b_i$ delivered over the time interval; that is, $\sigma_i=\int_{T_1}^{T_2}b_i(t)dt$. 

The idea is to assign full bandwidth to flows one at a time. Set $t_0=0$ and $t_i=\sum_{k=0}^i \frac{\sigma_k}{c_m}$, for $i\in [n]$. Define the new bandwidths as follows. For $i \in [n]$,
\[ b''_i(t) =
\begin{cases}
    c_m,		& \text{if } t \in (T_1 + t_{i-1},T_1 + t_i]\\
    0,     & \text{otherwise}\\
\end{cases}
\]

It is obvious that $b''_i(t)$ is feasible and can be non-zero only in the interval $(T_1 + t_{i-1}, T_1 + t_i]$. Furthermore, this schedule does not require any additional time. Indeed,
\begin{align*} T_1 + t_n\ &=\ T_1 + \sum_{k=0}^n \frac{\sigma_k}{c_m} =\ T_1 + \sum_{k=0}^n \frac{\int_{T_1}^{T_2}b_k(t)dt}{c_m}\\
&= T_1 + \int_{T_1}^{T_2} \sum_{k=0}^n \frac{b_k(t)}{c_m} dt\ \leq\ T_1 + \int_{T_1}^{T_2} dt\ =\ T_2,
\end{align*}
where the inequality follows from the requirement that $\forall t : \sum_k b_k(t) \leq c_m$.
\end{proof}

\subsection{Paths Are Given}\label{sec:circuit_based_paths_given}
In this setting, each connection request comes with a path. We give a $17.6$-approximation algorithm for this problem. We begin with an observation. 
\begin{observation}
The circuit-based coflow scheduling with given paths is strongly $\mathbf{NP}$-complete.
\end{observation}
This result follows directly from the $\mathbf{NP}$-completeness of $1 | pmtn, r_i | \sum \omega_i c_i$ which seeks to minimize total weighted completion time on a single machine with preemption and release times \cite{Labetoulle1984245}. Indeed, given an instance of the scheduling problem, consider a single edge $e:$
\begin{tikzpicture}[every node/.style={draw,shape=circle,fill=black,minimum size=1mm,inner sep=0pt,outer sep=0pt}]{
\node[label=left:$s$] (a) at (0,0) {};
\node[label=right:$t$] (b) at (1,0) {};
\draw[->,>=latex] (a) to (b);}
\end{tikzpicture}
such that every job $j$ with a processing time of $p_j$ corresponds to a coflow consisting of just one flow of size $\sigma_j=p_j$ from $s$ to $t$ and with the corresponding release time. Now, Lemma \ref{lem:disc_bandw} states that we can turn a solution returned by solving the circuit-based coflow problem into a preemptive schedule: simply determine the order of release and completion times on the time line and apply the lemma to each interval.

We follow our general framework for coflow scheduling.
In order to get the best approximation ratio, we optimize over some parameters such as interval length.

{\bf Reformulation}.
We transform the problem of minimizing total weighted coflow completion time to an instance of minimizing total weighted flow completion time. To capture completion times at the level of coflows, we (a) introduce a dummy flow $f^i_0$ in each coflow $F_i$, and (b) assign new weights $\omega'_{ij}$ to each flow. In more detail, the dummy flow comes with precedence constraints from $f^i_j, j\neq 0$ to $f^i_0$, requiring that in each coflow, $f^i_0$ finishes last. In effect, the precedence graph is a forest of $n$ depth-1 in-trees. We set the weight to $\omega_i$ for the dummy flow $f^i_0$  (i.e. $\omega'_{i0} = \omega_i$) and to zero for all other flows (i.e. $\omega'_{ij}=0$, $j\neq 0$). Finally, note that dummy flows do not have source, destination or size but are required to finish no sooner than other flows in the same coflow. Now, the new objective function to be minimized is $\sum_i\sum_j \omega'_{ij} c^i_j$.

{\bf The linear program}.
In this step, we devise an interval-indexed linear program. The time line is divided into segments $[0,1]$, $(1, 1+\epsilon]$, $(1+\epsilon,(1+\epsilon)^2],$ ..., $(\tau_{\ell}, \tau_{\ell+1}]$ for $\ell\in \{0,1, ..., L\}$, where $\tau_{\ell}=(1+\epsilon)^{\ell-1}$, $\tau_0=0$, $\epsilon>0$ will be determined later, and L is a sufficiently large integer. 

In the linear program below, $x^i_{j\ell}$ is the portion of connection request $f^i_j$ completed in the $\ell^{\textnormal{th}}$ interval $(\tau_{\ell}, \tau_{\ell+1}]$ and $b^i_{j\ell}$ is the bandwidth of $p^i_j$ (i.e. $f^i_j$'s path) during the $\ell^{\textnormal{th}}$ interval. Notice that release times are assumed to occur at the start of intervals; we show that this restriction costs us no more than a factor of $1+\epsilon$. Constraints (\ref{circ_prec_const_np}) ensure that the dummy connection request finishes last in each coflow. Constraints (\ref{eq:rate_np}) calculate the bandwidth of each path $p^i_j$ for every time interval (which is the fraction of flow size delivered, divided by interval length). Constraints (\ref{eq:flow_req_t_np}) ensure that no capacity violation occurs; note that $P(e)$ denotes the set of connection requests whose paths use edge $e$. Finally, constraints (\ref{eq:circuit_flow_release_time_np}) ensure that release times are respected.
\[ \text{Minimize}\ \sum_i \sum_j \omega'_{ij} c^i_j \;\;\;\;\;\;\;\;\;\;subject\ to\]
\begin{align}
\sum_{\ell\leq L} x^i_{j\ell} & = 1 &\forall& i,j \\
\sum_{\ell\leq L} \tau_{\ell} x^i_{j\ell} &\leq c^i_j &\forall& i,j \label{eq:mass_np}\\
c^i_j & \leq\  c^i_0 &\forall& i,j\label{circ_prec_const_np}\\
b^i_{j\ell} &= \sigma^i_j x^i_{j\ell} / \tau_{\ell}& \forall& i,j,\ell \label{eq:rate_np}\\
\sum_{f^i_j\in P(e)} b^i_{j\ell}  & \leq c(e) & \forall& \ell,e \label{eq:flow_req_t_np}\\
r^i_j > \tau_{\ell+1} &\Rightarrow x^i_{j\ell} = 0 & \forall& i, \ell \label{eq:circuit_flow_release_time_np}\\
x^i_{j\ell} &\geq 0 & \forall& i, j, \ell, e
\end{align}

\begin{lemma}\label{lem:circuit_based_coflow_mutiple_LP_lower}
$\frac{1}{1+\epsilon}$ times the optimal value of the above linear program is a lower bound on the weighted sum of completion times for circuit-based coflows with paths and release times.
\end{lemma}

\begin{proof}
Consider a schedule $\langle p^i_j, b^i_j(t) \rangle$ achieving average coflow completion time $\hat{C}$. We show that it corresponds to a feasible solution to the linear program with an objective value $\leq (1+\epsilon)\hat{C}$. 
Start by setting the variables $x^i_{j\ell}$ for each connection request $f^i_j$ according to the interval in which they finish. Next, note that release times are arbitrary; hence, a connection request may arrive in the middle of an interval. On the other hand, constraints (\ref{eq:rate_np}) require constant bandwidth in each interval. To remedy this, move all release times to the end of the interval in which they occur and update $x^i_{j\ell}$ if necessary; this clearly increases $\hat{C}$ by no more than a factor of $1+\epsilon$.

Finally, for each interval $(\tau_{\ell}, \tau_{\ell+1}]$, apply Lemma \ref{lem:const_bandw} to get constant bandwidth for each connection request. More precisely, we set
\[ b^i_{j\ell} = \frac{\int_{\tau_{\ell}}^{\tau_{\ell+1}} b^i_j(t)dt}{\tau_{\ell}}\]
for each $f^i_j$. Clearly, constraints (\ref{eq:rate_np}) are satisfied; furthermore, the lemma ensures that constraints (\ref{eq:flow_req_t_np}) stay satisfied as well. This setting of variables satisfies all the remaining constraints as well.
\end{proof}

{\bf Rounding}.
The next step after solving the linear program is to construct a schedule. Let $\{\hat{x}^i_{j\ell}\}_{ij\ell}$ denote an optimal solution to the LP. Set $\hat{b}^i_{j\ell}$ according to equation (\ref{eq:rate_np}).
For some $0<\alpha\leq 1$, define the $\alpha$-interval $h^{\alpha}_{ij}$ of a connection request $f^i_j$ to be the interval $(\tau_{h^{\alpha}_{ij}},\tau_{h^{\alpha}_{ij}+1}]$ which contains the earliest point in time that a cumulative $\alpha$-fraction of the connection request has been completed; that is
\[h^{\alpha}_{ij} = \min \left\{\ell : \sum_{t=0}^{\ell} x^i_{jt} \geq \alpha \right\}.\]
The exact value of $\alpha$ will be determined later.

In the schedule, each connection request $f^i_j$ will run entirely in $(\tau_{h^{\alpha}_{ij}+D}, \tau_{h^{\alpha}_{ij}+D+1}]$--that is, the $D^{\textnormal{th}}$ interval after its $\alpha$-interval. The parameter $D\in \mathbb{N}$, the {\em displacement factor}, will be determined later. Note that the length of this interval is $(1+\epsilon)^{D-1}$ times the cumulative length of intervals from the beginning to the end of the $\alpha$-interval.

Fix a $k$. Let $S[k]$ denote the set of connection requests scheduled to run in $(\tau_k, \tau_{k+1}]$. Clearly, the $\alpha$-interval of all these connection requests is $k-D$. Consider a connection request $f^i_j \in S[k]$. According to the optimal solution obtained by solving the LP, its bandwidth function is
\[
     \hat{b}^i_j(t)= 
\begin{cases}
    \hat{b}^i_{j0},		& \text{if } t \in [0,1],\\
    \hat{b}^i_{j1},     & \text{if } t \in (\tau_1, \tau_2],\\
    \;\,\vdots			&\;\;\;\,\vdots\\
    \hat{b}^i_{jk}, & \text{if } t \in (\tau_{k-D},\tau_{k-D+1}].
\end{cases}
\]
We now define a new constant bandwidth
\begin{equation}
\tilde{b}^i_j= \sum_{\ell=0}^{\ell=k-D} \frac{\alpha^{-1}\hat{b}^i_{j\ell}\tau_{\ell}}{\tau_{k}}= \sum_{\ell=0}^{\ell=k-D} \frac{\alpha^{-1}\hat{b}^i_{j\ell}}{(1+\epsilon)^{k-\ell}}.
\end{equation}
Observe that we have scaled and added up each bandwidth value depending on how far its interval is from the $k^{\textnormal{th}}$ interval. 
For each connection request in $S[k]$, at least an $\alpha$ fraction of its size $\sigma^i_j$ was delivered by the end of the $\alpha$-interval. But now, with the new bandwidth, the entire flow is delivered over the stretched time line. More precisely,
\[ \forall f^i_j \in S[k] : \int_{0}^{\tau_{h^{\alpha}_{ij}+1}} \hat{b}^i_j(t) dt = \int_{0}^{\tau_{k-D+1}} \hat{b}^i_j(t) dt = \alpha\tilde{b}^i_j \tau_{k}.\]
This also means that the new bandwidth and the displacement factor have to satisfy
\begin{equation}\label{eq:displacement}
D\geq \Big\lceil \log_{1+\epsilon}\frac{1}{\alpha} \Big\rceil + 1.
\end{equation}
Furthermore, we must ensure that no capacity violation occurs. In the schedule, for every $k$ and for every edge $e$,
\begin{align*}
\sum_{f^i_j\in P(e)} \tilde{b}^i_j &=\sum_{f^i_j\in P(e)}\sum_{\ell=0}^{\ell=k-D} \frac{\alpha^{-1}\hat{b}^i_{j\ell}}{(1+\epsilon)^{k-\ell}}\\
&= \sum_{\ell=0}^{\ell=k-D} \frac{\alpha^{-1}}{(1+\epsilon)^{k-\ell}} \sum_{f^i_j\in P(e)} \hat{b}^i_{j\ell}\\
&\leq c(e) \sum_{\ell=0}^{\ell=k-D} \frac{\alpha^{-1}}{(1+\epsilon)^{k-\ell}},
\end{align*}
where the inequality follows from constraints (\ref{eq:flow_req_t_np}). Thus, we require that
$ \sum_{\ell=0}^{\ell=k-D} \frac{\alpha^{-1}}{(1+\epsilon)^{k-\ell}} \leq 1$, for all $k$. Hence,
\begin{equation}\label{eq:capacity_constr}
\frac{1}{\epsilon(1+\epsilon)^{D-1}} \leq \alpha.
\end{equation}
Finally, note that the blow up factor in completion time for each connection request in $S[k]$, taking into account the $(1+\epsilon)$-factor loss in the LP, is at most
\begin{equation}\label{eq:apprx_ratio}
\frac{(1+\epsilon)\tau_{k+1}}{(1-\alpha)\tau_{k-D}}=
\frac{(1+\epsilon)^{D+2}}{1-\alpha}.
\end{equation}
Numerically minimizing (\ref{eq:apprx_ratio}) subject to (\ref{eq:capacity_constr}) and displacement factor constraint  (\ref{eq:displacement}), we get an approximation factor of $17.5319$ for $\alpha = 0.5$, $D=3$, and $\epsilon\approx 0.5436$.

\subsection{Paths Are Not Given}\label{sec:circuit_based_paths_not_given}
In this section, we consider circuit-based coflow scheduling when paths are not given. Our $O(\frac{\log |E|}{\log\log |E|})$ approximation ratio for this case is asymptotically tight since an $\Omega(\frac{\log |E|}{\log\log |E|})$ hardness follows from hardness of congestion minimization on directed graphs \cite{DBLP:conf/stoc/ChuzhoyGKT07}. Recall that, in the congestion minimization problem, given a directed graph and source-sink pairs, the goal is to connect each pair by a path such that congestion, the maximum number of paths crossing an edge, is minimized. But, regarding the source-sink pairs as connection requests of unit size in a single coflow, the congestion is equivalent to makespan, where all edge capacities are set to 1. So, essentially, congestion minimization can be seen as a special case of our problem. 

We now go over our algorithm (see Algorithm \ref{algorithm}).

\begin{algorithm} \label{algorithm}
 \SetNoFillComment \SetKwInOut{Input}{input}
 \SetKwInOut{Output}{output}
 \SetKwInOut{Return}{return}
 \SetKwFor{ForEach}{foreach}{do}{\nl end}
  \Input{Coflows $\mathcal{F}$, Network $(G,c)$}
  \nl Construct the LP $\mathcal{L}$\;
  \nl Solve $\mathcal{L}$; perform rounding and scaling to get flow values $x^f_e$ and completion times $c_f$\;
  \nl\ForEach{flow $f \in \mathcal{F}$}{
   \nl PathSet[$f$] = FlowDecomposition($f$)\;
   \nl Path[$f$] = Rounding(PathSet[$f$])\;}
   \Return{Flow paths and ordering based on $c_f$}
  \caption{{\sc Circuit-based coflow scheduling}}
 \end{algorithm}


%
%
%
{\bf Reformulation}.
Identical to \S\ref{sec:circuit_based_paths_given}.

{\bf The linear program}.
The linear program is similar to \S\ref{sec:circuit_based_paths_given}, where we have set $\epsilon=1$.
Notice that constraints (\ref{eq:flow_req_s})-(\ref{eq:flow_req_t}) are the usual flow constraints per time interval. 
\begin{align}
\text{Minimize}\ &\sum_i \sum_j \omega'_{ij} c^i_j & subject\ to\nonumber \\
\sum_{\ell\leq L} x^i_{j\ell} & = 1 &\forall& i,j \\
\sum_{\ell\leq L} \tau_{\ell} x^i_{j\ell} &\leq c^i_j &\forall& i,j \label{eq:mass}\\
c^i_j & \leq\  c^i_0 &\forall& i,j\label{circ_prec_const}\\
\sum_{\substack{e\in N(v) \\ v\not\in \{s^i_j,d^i_j\}}} x^{e}_{ij\ell} &= 0& \forall& i,j,\ell \label{eq:flow_req_s}\\
\sum_{e\in N(d^i_j)} x^{e}_{ij\ell} &= \sigma^i_j x^i_{j\ell} / \tau_{\ell}& \forall& i,j,\ell \label{eq:rate}\\
\sum_{e\in N(s^i_j)} x^{e}_{ij\ell} &= -\sum_{e\in N(d^i_j)} x^{e}_{ij\ell}& \forall& i,j,\ell\\
\sum_{i} \sum_{j} x^{e}_{ij\ell} & \leq c(e) & \forall& \ell,e \label{eq:flow_req_t}\\
r^i_j > \tau_{\ell+1} &\Rightarrow x^i_{j\ell} = 0 & \forall& i, j, \ell \label{eq:circuit_flow_release_time}\\
x_{ij\ell}^e, x^i_{j\ell} &\geq 0 & \forall& i, j, \ell, e
\end{align}

Next, we show that the optimal value of the linear program gives us a relaxed lower bound.

\begin{lemma}\label{lem:circ_path_not_given_lp_lower_bound}
Half the optimal value of the above linear program provides a lower bound on the weighted sum of completion times for circuit-based coflows with release times.
\end{lemma}

\begin{proof}
Following the proof of Lemma \ref{lem:circuit_based_coflow_mutiple_LP_lower}, move release times if necessary, set the variables $x^i_{j\ell}$ and calculate bandwidths per interval. Next, set the value of flow variables $x^e_{ij\ell}$ based on the paths $p^i_j$ and the calculated bandwidths. This setting of variables satisfies all constraints.
\end{proof}

{\bf Rounding}.
This step involves scaling followed by flow decomposition and randomized rounding. Let $\langle \hat{x}_{ij\ell}^e, \hat{x}^i_{\ell} \rangle$ denote an optimal solution to the LP. Analogous to the previous case in \S\ref{sec:circuit_based_paths_given}, define the half interval $h^i_j$ of each connection request $f^i_j$. All connection requests whose half interval is $H$ will be scheduled to run entirely in interval $(\tau_{H+3}, \tau_{H+4}]$. In analogy with \S\ref{sec:circuit_based_paths_given}, set $\alpha=\frac{1}{2}$ and $D=3$. Now, fix a $k$ and let $S[k]$ denote the set of connection requests scheduled to run in $(\tau_k, \tau_{k+1}]$. Perform scaling and adding up of flows for every edge to get updated flow variables
\begin{equation} 
\tilde{x}^e_{ij} = \sum_{\ell=0}^{\ell=k-3} \frac{\hat{x}_{ij\ell}^e}{2^{k-\ell-1}}.
\end{equation}
The calculations in \S\ref{sec:circuit_based_paths_given} apply here as well (where we use inequalities (\ref{eq:flow_req_t}) instead of (\ref{eq:flow_req_t_np})). Hence, the entire size of every connection request is delivered in the new settings and all edge capacities are respected.
%
%
The next step is to apply the well-known flow decomposition theorem (see e.g. \cite{Ahuja:1993:NFT:137406}) to the flow variables $\tilde{x}_{ij}^e$ to get a set of flow paths $P_{ij}=\{p^i_{j}\}$ from $s^i_j$ to $d^i_j$ for each connection request $f^i_j\in S[k]$.

At this point, after computing the flow paths, we have for each $k$ a set of connection requests $S[k]$ expressed in terms of flow paths. Unfortunately, this is not enough since each connection request must have exactly one path from source to sink. What we have, instead, is a \emph{set} of flow paths $P_{ij}=\{p^i_{j1},...,p^i_{jn_{ij}}\}$ for each connection request $f^i_j$. This brings us to the final step.

Let $b^i_j$ denote the total amount of flow in $P_{ij}$;
\[ b^i_j = \sum_{k\in [n_{ij}]} |p^i_{jk}|. \]
We apply the randomized rounding technique of Raghavan and Thompson \cite{raghavan+t:round} to choose a single path from $P_{ij}$: we randomly choose a flow path from $P_{ij}$ where the probability of $p^i_{jk}$ being selected is $|p^i_{jk}|/b^i_j$. All of the flow in $P_{ij}$ will go over the chosen path.
Let $\bar{b}_e$ denote the total flow passing through $e$ after the rounding. Using a Chernoff-Hoeffding bound, we can show that for every $e\in E$ and every $\delta > 0$:
\junk{for any $0< \epsilon <1$,
\[\Pr[\max_e \frac{\bar{b}_e}{c(e)} \geq (1+\delta)] \leq \epsilon \]
where $\delta=\frac{\epsilon\log|E|}{\log\log|E|}$.

All of this is pretty standard. However, for the sake of completeness, we give more details. Note that, it follows from the Chernoff-Hoeffding bound that, for every edge $e\in E$ and every $\delta > 0$:}
\begin{equation*}
\Pr\left[\bar{b}_e \geq (1+\delta) c(e)\right] \leq \Big( \frac{\mathbf{e^{\delta}}}{(1+\delta)^{1+\delta}} \Big)^{c(e)/|E|}.
\end{equation*}
Hence,
\begin{align*}
\Pr\left[\max_e \frac{\bar{b}_e}{c(e)} \geq (1+\delta)\right] &\leq \sum_e \Pr\left[\frac{\bar{b}_e}{c(e)} \geq (1+\delta)\right]\\
&\leq \sum_e \Big( \frac{\mathbf{e^{\delta}}}{(1+\delta)^{1+\delta}} \Big)^{c(e)/|E|}\\
&\leq |E| \Big( \frac{\mathbf{e^{\delta}}}{(1+\delta)^{1+\delta}} \Big)^{(\min_e c(e))/|E|}.
\end{align*}
If we set $\delta = O(\frac{\log |E|}{\log\log |E|})$, we obtain that the bandwidth used exceeds the capacity of any edge by at most $O(\frac{\log |E|}{\log\log |E|})$ with probability at most $1/n^{c}$ for some constant $c > 0$.  
Consequently, by scaling down the bandwidth allocated by $1 + \delta$, and scaling up in time by the same factor, we obtain a feasible set of connection requests with a blowup in completion times at most $O(\frac{\log |E|}{\log\log |E|})$.


\section{Packet-based coflows}\label{sec:packet_based}
In this model, a flow corresponds to a {\em single} packet (flow size is 1) and a coflow is a collection of packets.
The schedule is based on discrete time steps. During a step, each packet can either stay at a node (in a queue) or move along an edge with the restriction that at most one packet can use the edge at a time. On the hardness side, packet scheduling can not be approximated within $6/5-\epsilon$ , for all $\epsilon >0$, unless $\mathbf{P}=\mathbf{NP}$ \cite{DBLP:conf/waoa/PeisSW09}. This is true even if paths are given.
%
In \S\ref{sec:packet_based_paths_given}, we observe that when paths are given, we can apply known results for job shop scheduling. Next, in \S\ref{sec:packet_based_paths_not_given}, we consider the case where paths are not given. We present an $O(1)$-approximation algorithm for this case.  

\subsection{Paths Are Given}\label{sec:packet_based_paths_given}
When a path $p^i_j$ is given for each packet $f^i_j$, the only remaining problem is scheduling. It is well-known that the problem of scheduling packets on a store-and-forward network with fixed paths is an instance of job-shop scheduling. Specifically, each packet $f^i_j$ can be regarded as a job which needs to be executed on a number of machines (edges) in a prescribed order (given by $p^i_j$) such that a machine can process one job at a time. Consequently, the packet-based coflow problem with given paths, is an instance of shop scheduling $J |r_j, p_{ij}=1| \sum_S \omega_S C_S$ with unit processing times and a \emph{generalized} min-sum objective; hence, the algorithm of \cite{Queyranne02approximationalgorithms} immediately gives an $O(1)$-approximation. 

\begin{theorem}[\cite{Queyranne02approximationalgorithms}]
\label{t:qs}
There exists a polynomial-time $O(1)$-approximation algorithm for the (generalized) min-sum job-shop problem $J(P) |r_j, p_{ij}=1| \sum_j \omega_j C_j$.
\end{theorem}

\subsection{Paths Are Not Given}\label{sec:packet_based_paths_not_given}
This is the more general case where one needs to both route and schedule packets. We give a polynomial-time $O(1)$-approximation algorithm for this problem. 

An ingredient in our solution is the notion of \emph{time-expanded graphs} first introduced by Ford and Fulkerson \cite{doi:10.1287/opre.6.3.419}. In more detail, given a (directed) graph $G=(V,E)$ and $n$ coflows $\{F_i=\{f^i_1,..., f^i_{n_i}\}\}_{i\in [n]}$, we construct a time-expanded graph $G^T$ as follows. Each node in the time-expanded graph is labeled by a tuple $(v, t)$ where $v$ is a node in $G$ and $0\leq t\leq T$ is the timestamp. Therefore, the edge set of $G^T$ is 
\begin{align*}
E(G^T)=\ &\{((u,t), (v,t+1)) : (u,v) \in E(G),\ 0\leq t\leq T-1\}\\
&\cup \{((v, t), (v,t+1)) : v \in V(G),\ 0\leq t\leq T-1\}
\end{align*}
where the second group of edges are called queue edges and their role is to simulate packets waiting for one or more rounds at a node (see Figure \ref{fig:time_exp_graph}). 
Let us denote by $E'$ the set $E(G^T)$ excluding the queue edges.

The overall idea of the algorithm is as follows. By a net-flow we mean a network flow in the ordinary sense. We express each packet as multiple net-flows to the packet's destination in different times in $G^T$; this will become clear shortly. The remaining steps involve writing an interval-indexed linear program, assigning packets to intervals, and performing routing and scheduling packets for each interval.

\begin{figure}
\begin{center}
\begin{tikzpicture}[scale=0.85]
\tikzstyle{every node}=[font=\small]

\node[draw,circle,fill=orange] (s) at (4,2.5) {};
\draw (4.25,2.5) node[right] {$s^i_j$};
\node[draw,circle] (a) at (3,3.5) {};
\node[draw,circle] (b) at (5,3.5) {};
\draw (5.25,3.5) node[right] {$\mathbf{G}$};
\node[draw,circle,fill=cyan] (t) at (4,4.5) {};
\draw (4.25,4.5) node[right] {$d^i_j$};
\draw [->,>=latex] (s)--(a);
\draw [->,>=latex] (a)--(t);
\draw [->,>=latex] (t)--(b);
\draw [<-,>=latex] (a)--(b);

\node[draw,circle,fill=orange] (s0) at (1,-1) {};
\draw (1.25,-1) node[right] {$(s^i_j,0)$};
\node[draw,circle] (a0) at (0,0) {};
\node[draw,circle] (b0) at (2,0) {};
\node[draw,circle] (t0) at (1,1) {};
\draw (1.25,1) node[right] {$(d^i_j,0)$};

\node[draw,circle] (s1) at (4,-1) {};
\draw (4.25,-1) node[right] {$(s^i_j,1)$};
\node[draw,circle] (a1) at (3,0) {};
\node[draw,circle] (b1) at (5,0) {};
\node[draw,circle,fill=cyan] (t1) at (4,1) {};
\draw (4.25,1) node[right] {$(d^i_j,1)$};

\node[draw, circle] (s2) at (7,-1) {};
\draw (7.25,-1) node[right] {$(s^i_j,2)$};
\node[draw, circle] (a2) at (6,0) {};
\node[draw, circle] (b2) at (8,0) {};
\node[draw,circle,fill=cyan] (t2) at (7,1) {};
\draw (7.25,1) node[right] {$(d^i_j,2)$};

\draw [->,>=latex] (s0)--(a1);
\draw [->,>=latex] (a0)--(t1);
\draw [->,>=latex] (t0)--(b1);
\draw [->,>=latex] (b0)--(a1);
\draw [->,>=latex] (s1)--(a2);
\draw [->,>=latex] (a1)--(t2);
\draw [->,>=latex] (t1)--(b2);
\draw [->,>=latex] (b1)--(a2);

\draw[dashed,->] (s0) to [bend left] (s1);
\draw[dashed,->] (t0) to [bend left] (t1);
\draw[dashed,->] (a0) to [bend left] (a1);
\draw[dashed,->] (b0) to [bend left] (b1);

\draw[dashed,->] (s1) to [bend left] (s2);
\draw[dashed,->] (t1) to [bend left] (t2);
\draw[dashed,->] (a1) to [bend left] (a2);
\draw[dashed,->] (b1) to [bend left] (b2);

\end{tikzpicture}
\end{center}
\caption{An example graph $G$ (above) and its time-expanded version $G^T$ for $T=2$ (below). Packet $f^i_j$ needs to be routed from node $s^i_j$ to node $d^i_j$ in $G$. Corresponding to $f^i_j$, flows of combined size 1 are sent from $(s^i_j,0)$ to $(d^i_j,k)$ for $k=1, 2$. Dashed lines correspond to queue edges.}\label{fig:time_exp_graph}
\end{figure}

{\bf Reformulation}.
Consider a packet $f^i_j \in F_i$ to be routed from $s^i_j$ to $d^i_j$ in $G$.  Corresponding to this packet, we send a set of net-flows $f^t_{ij}$ in $G^T$: net-flow $f^t_{ij}$ has node $(s^i_j, r^i_j)$ as source (demand is -1) and node $(d^i_j, t)$ as destination where $t\in \{r_i+1, ..., T\}$ and sum of demands of these nodes is 1 (the exact values will be determined by solving the linear program below). Also, we set $T$ to a sufficiently large integer (say $|E|\sum_{i\in [n]} n_i$). 
As before, we introduce dummy flows $f^i_0$ and update weights $\omega'_{ij}$ accordingly.

{\bf The linear program}.
Let $T'=\log T$, $\tau_0=1$ and $\tau_\ell = 2^{\ell -1}$ for $\ell \in [T']$. Let $c^i_j$ denote the earliest time when a net-flow $f^t_{ij}$ in $G^T$ reaches its destination; that is $c^i_j=\min_t c^t_{ij}$ where $c^t_{ij}$ is the completion time of $f^t_{ij}$. 
\begin{align}
\text{Minimize}\ \sum_{i\in [n]}&\sum_{j = 0}^{n_i} \omega'_{ij} c^i_j &subject\ to\;\,& \nonumber \\
\sum_{\tau_{\ell} < t \leq \tau_{\ell+1}}\sum_{e \in E((d_{ij},t))} x_{e}^{f^t_{ij}} & = f_{ij\ell} &\forall i,j,\ell \label{eq:flow_LP_s} \\
\sum_{\ell\leq T'} \tau_{\ell} f_{ij\ell} &\leq c_{ij} & \forall i, j \label{eq:cij}\\
c_{ij} & \leq c_{i0} & \forall i,j \label{prec_const2}\\
\sum_{i\in [n]} \sum_{j=0}^{n_i}\sum_{t\leq \tau_{\ell+1}} x_{e}^{f^t_{ij}} & \leq \tau_{\ell+1} & \forall \ell,e \label{eq:flow_congestion} \\
\sum_{e\in E'} \sum_{t\leq \tau_{\ell+1}} x_{e}^{f^t_{ij}} & \leq \tau_{\ell+1} & \forall i,j,\ell \label{eq:flow_dilation} \\
\sum_{t\in T} b^t_{ij} &= 1 & \forall i, j \\
\mathcal{N}^{f^t_{ij}} x^{f^t_{ij}} &= b^{f^t_{ij}} & \forall i,j \label{eq:flow_reqF} \\
0 \leq x_{e}^{f^t_{ij}} &\leq 1 & \forall i , j, t, e \label{eq:flow_LP_e}
\end{align}

Variables $x_{e}^{f^t_{ij}}$ denote the amount of flow corresponding to $f_{ij}^t$ over edge $e\in E(G^T)$. In constraints (\ref{eq:flow_LP_s}), $f_{ij\ell}$ is the total amount of flow, corresponding to net-flows $f^t_{ij}$, that enter their destination nodes $(d^i_j, t)$ for $t\in (\tau_{\ell}, \tau_{\ell+1}]$. By $E((d_{ij},t))$ we mean the set of incoming edges to $(d_{ij},t)$ in constraints (\ref{eq:flow_LP_s}). Moreover, $b^t_{ij}$ is the demand of node $(d^i_j, t)$ and $\mathcal{N}^{f^t_{ij}} x^{f^t_{ij}} = b^{f^t_{ij}}$ specify the obvious requirements of sending a net-flow of size $b^t_{ij}$ from $(s^i_j,r_i)$ to $(d^i_j, t)$ for $r_i\leq t\leq T$, as well as flow conservation at nodes other than source and sink. Constraints (\ref{eq:flow_congestion}) and (\ref{eq:flow_dilation}) restrict respectively the congestion and dilation of packets which reach their destination in interval $[0, \tau_{\ell+1}]$. As can be seen, queue edges don't contribute to path dilation.

\begin{lemma}
The optimal value of the linear program (\ref{eq:flow_LP_s})-(\ref{eq:flow_LP_e}) with weights $\omega'_{i0}=\omega_i$ and $\omega'_{ij}=0$, $j \neq 0$, is a lower bound on the optimal weighted sum of packet-based coflow completion times.
\end{lemma}

\begin{proof}
Given a set $S=\langle \{p^i_j\}, \{q^i_j\} \rangle_{i\in [n], j\in [n_i]}$ of paths $\{p^i_j\}$ and schedules $\{q^i_j\}$, we show that $S$ constitutes a feasible solution to the LP. Without loss of generality, let $g$ denote the maximum amount of time any packet has to wait in the schedule; the existence of $g\in O(1)$ follows from \cite{Leighton94packetrouting} (note that $g$ is the maximum queue size.) Each packet $f^i_j$ is routed, on a path of length $|p^i_j|$, from $s^i_j \in V$ to $d^i_j \in V$ according to the schedule $\{q^i_j\}$ which specifies the amount of time $\in [0, g]$ the packet has to wait before crossing the next edge on its path. Let $T^i_j$ denote the time step when $f^i_j$ reaches its sink. Furthermore, we define $T=\max_{ij} T^i_j$.

For each packet $f^i_j$, we construct an integral flow $f^{T_{ij}}_{ij}$ in the time-expanded graph $G^T$. Consider a packet $f^i_j$ with release time $r^i_j$ in $G$ going from $s^i_j$ to $d^i_j$ on a path specified by $p^i_j$ according to the schedule $q^i_j$ and arriving at time $T^i_j$. The corresponding flow, $f^{T_{ij}}_{ij}$, in $G^T$ from $(s^i_j,r^i_j)$ to $(d^i_j, T^i_j)$ is constructed in steps. We keep a current time step $t$ which is initially set to $r^i_j$. For each edge $e=(u,v)$ on the path $p^i_j$, starting from $s^i_j$, the flow $f^i_j$ follows the edge(s)
\[
((u,t), (u,t+1)) \rightarrow  \cdots \rightarrow  ((u,t+h), (v,t+h+1)),
\]
in $G^T$ where $h$ is the delay specified by $q^i_j$ before crossing the edge. Next, we update $t\leftarrow t+1$ and repeat until we reach $d^i_j$. Clearly, $T^i_j$ is the sum of the path length $|p^i_j|$ and total amount of delay prescribed by the schedule $q^i_j$.

It follows from the construction that constraints (\ref{eq:flow_reqF}) are satisfied. Furthermore, since the schedule does not allow more than one packet to cross an edge at the same time, the path followed by a flow is exclusive to that flow except, possibly, the edges $((u,t),(u,t+1))$ which effectively simulate queues. At this point, we can set the value of variables $x_{e}^{f^t_{ij}}$ according to the flow $f^{T^i_j}_{ij}$. We also set $c^i_j=T^i_j$ and $c^i_0=\max_{j\in [n_i]} c^i_j$. It is easy to verify that constraints (\ref{eq:flow_LP_s})-(\ref{prec_const2}) are satisfied. Regarding congestion constraints (\ref{eq:flow_congestion}), observe that, for all $\ell$, packets that have reached their destination by time $\tau_{\ell+1}$, can not cause an edge to have a congestion value more than $\tau_{\ell+1}$. The same argument applies to dilation constraints (\ref{eq:flow_dilation}).
\end{proof}

{\bf Rounding}.
At a high-level, this step consists  of two parts (a) Assign the packets to time intervals; (b) Route and schedule packets in each interval by applying the algorithm of Srinivasan and Teo \cite{TeoSrinv}.

Let $\hat{S}=\langle \hat{x}_{e}^{f^t_{ij}}, \hat{f}_{ij\ell}, \hat{c}^i_j \rangle$ denote an optimal solution to the LP (\ref{eq:flow_LP_s})-(\ref{eq:flow_LP_e}). We start by filtering the solution and assigning packets to intervals. Specifically, define the half-interval of a packet $f^i_j$ to be the interval $(\tau_{h^i_j}, \tau_{h^i_j+1}]$ such that $h^i_j=\min \big\{\ell : \sum_{t < \ell} \hat{f}_{ijt} \geq \frac{1}{2} \big\}$;
then, set
\[
     \bar{f}_{ij\ell}= 
\begin{cases}
    \hat{f}_{ij\ell}/(1-\sum_{t\geq\ell} \hat{f}_{ijt}),& \text{if } \tau_{\ell+1} \leq h^i_j,\\
    0,              & \text{if } \tau_{\ell+1} \geq h^i_j.
\end{cases}
\]
Also, set the other variables $\bar{x}_{e}^{f^t_{ij}}$ and $\bar{c}^i_j$ according to (\ref{eq:flow_LP_s}) and (\ref{eq:cij}) respectively. Clearly, the blowup in each variable and the right hand side of (\ref{eq:flow_congestion})-(\ref{eq:flow_dilation}) is at most 2.

Next, we assign each packet $f^i_j$ to its half-interval. Let $P[\ell]$ denote the set of packets that have been assigned to the interval $(\tau_{\ell}, \tau_{\ell+1}]$. At this stage, we go over all time intervals sequentially and, for each one, route and schedule all packets in it. Fix an $\ell$ and consider all packets $f^i_j \in P[\ell]$. Collapse the portion of the time-expanded graph $G^T$, corresponding to the times $t\leq \tau_{\ell+1}$, back to $G$ by combining nodes and edges which differ only in time stamp and by removing queue edges altogether.

Observe that the removal of queue edges will not cause any problem since any amount of flow passing through them will have left the queue by the end of the last interval. 
Now, the filtered solution $\bar{S}=\langle \bar{x}_{e}^{f^t_{ij}}, \bar{f}_{ij\ell}, \bar{c}^i_j \rangle$ satisfies
\begin{align}
\sum_{i, j : f^i_j \in F[\ell]} x_e^{f_{ij}} & \leq \tau_{\ell+2} & \forall& e \in E \label{eq:flow_LP2_s}\\
\sum_{e\in E} x_e^{f_{ij}} & \leq \tau_{\ell+2} & \forall& f^i_j \in P[\ell]\\
\bar{\mathcal{N}}^{f_{ij}} x^{f_{ij}} &= \bar{b}^{f_{ij}} & \forall& f^i_j \in P[\ell] \\
0 \leq x_e^{f_{ij}} &\leq 1 & \forall& f^i_j \in P[\ell]\ \forall e \in E \label{eq:flow_LP2_e}
\end{align}
where the constraints $\bar{\mathcal{N}}^{f_{ij}} x^{f_{ij}} = \bar{b}^{f_{ij}}$ express the flow requirements in the collapsed graph. Formally,
\begin{lemma}\label{lem:collapse}
For every $\ell \in [T']$, the set of flows $\bar{x}_e^{f_{ij}}$ corresponding to packets $f^i_j$ in $P[\ell]$ satisfy the LP in (\ref{eq:flow_LP2_s})-(\ref{eq:flow_LP2_e}) with respect to the (collapsed) graph $G$.
\end{lemma}

This LP is exactly the one considered in \cite{TeoSrinv} and by applying their theorem, we can route and schedule all packets in $P[\ell]$ such that the last packet arrives in time $O(\tau_{\ell+2})$.

\begin{theorem}[\cite{TeoSrinv}]
There are constants $c', c'' >0$ such that the following holds. For any packet routing problem on any network, there is a set of paths and a corresponding schedule that can be constructed in polynomial time such that the routing time is at most $c'$ times the optimal, and the maximum queue size at each edge is bounded by $c''$.
\end{theorem}

We do this for $P[1]$, $P[2]$ and so on. Note that the coflow release times are respected. The optimal completion time of a packet $f^i_j \in P[\ell]$ in the LP (\ref{eq:flow_LP_s})-(\ref{eq:flow_LP_e}) is $\hat{c}^i_j = O(\tau_{\ell+1})$. On the other hand, the completion time of the same packet in our algorithm is
\begin{equation}
\tilde{c}^i_j = \sum_{t\leq \ell} O(\tau_{t+1}) = \sum_{t\leq \ell} O(2^t) = O(2^{\ell}) = O(\tau_{\ell+1}).
\end{equation}
Consequently, the weighted completion time of our algorithm is within a $O(1)$ factor of the optimal solution.

\section{Experiments}\label{sec:experiments}
\noindent In this section, we describe some experiments conducted to evaluate the practical performance of our circuit-based scheduling algorithm in  \S\ref{sec:circuit_based_paths_not_given}.

\subsection{Methodology}\label{sec:expr_method}
We use a simulator to assess the performance of our algorithm in practice. Packet-level simulators are not suitable for this setting due to their high level of overhead \cite{Al-Fares:2010:HDF:1855711.1855730,rapier}. Therefore, like previous works \cite{Al-Fares:2010:HDF:1855711.1855730,rapier,ChowdhuryZhongStoica2014}, we developed a flow-based simulator. At a high level, the simulator is an event queue. Each flow corresponds to an event which happens at its release time. The simulator chooses the next flow based on the ordering prescribed by a scheduling algorithm or scheme. A second event occurs when a flow completes; at which time, its reserved bandwidth is released.

We use a 128-server Fat-Tree with 1Gbs links as the network topology. Due to the complexity of the linear program to be solved, simulations of large instances were prohibitively slow.
Each coflow instance is randomly generated with flow release times, flow sizes, and coflow weights based on Poisson distributions. Each result is the average of 10 tries.

\subsection{Implementation}
We implement Algorithm \ref{algorithm} in ~\S\ref{sec:circuit_based_paths_not_given} with some minor tweaks. In particular, to avoid wasting time and bandwidth, each flow starts as soon as it can (in the order prescribed by the linear program), as opposed to starting at the beginning of its half-interval. 

The path decomposition algorithm tries to minimize the number of paths per flow by finding the ``thickest'' paths; this is done using a well-known version of Dijkstra's shortest-path algorithm. 
We implemented the above algorithm in C++ using Lemon libraries and IBM CPLEX 12.6.3. to solve the linear program.

\subsection{Performance}
We measure the weighted sum of completion times. As in~\cite{rapier}, we compare our algorithm to the following schemes. 
\begin{itemize} 
\item Baseline: flows are routed and ordered randomly.
\item Schedule-only: flows are routed randomly; ordering is by minimum completion time which is computed as the ratio of flow size to path bandwidth.
\item Route-only: flows are routed for achieving good load balance and edge utilization; ordering is arbitrary.
\end{itemize}

\junk{Furthermore, we compare our results to the optimal value of the linear program (\ref{lp:first})-(\ref{lp:last}) which provides a lower bound on the total weighted completion time.}

We investigate the impact of two parameters. First is the number of flows in each coflow, which we refer to as {\em coflow width}. Second is the number of coflows.

\subsubsection*{Coflow width}
In this case, we fix the number of coflows to 10 and run experiments for coflow widths in $\{4,8,16,32\}$.
As illustrated in Figure~\ref{fig:changeCoflowWidth}, {\sc LP-Based} completes {\em every instance}\/ faster than the baseline, Schedule-only, and Route-only schedules.  On average, the improvement over the baseline, Schedule-only and Route-only is by $\%126$, $\%96$, and $\%22$, respectively.
\begin{figure}[h]
\begin{center}
\centering
\includegraphics[scale=0.3]{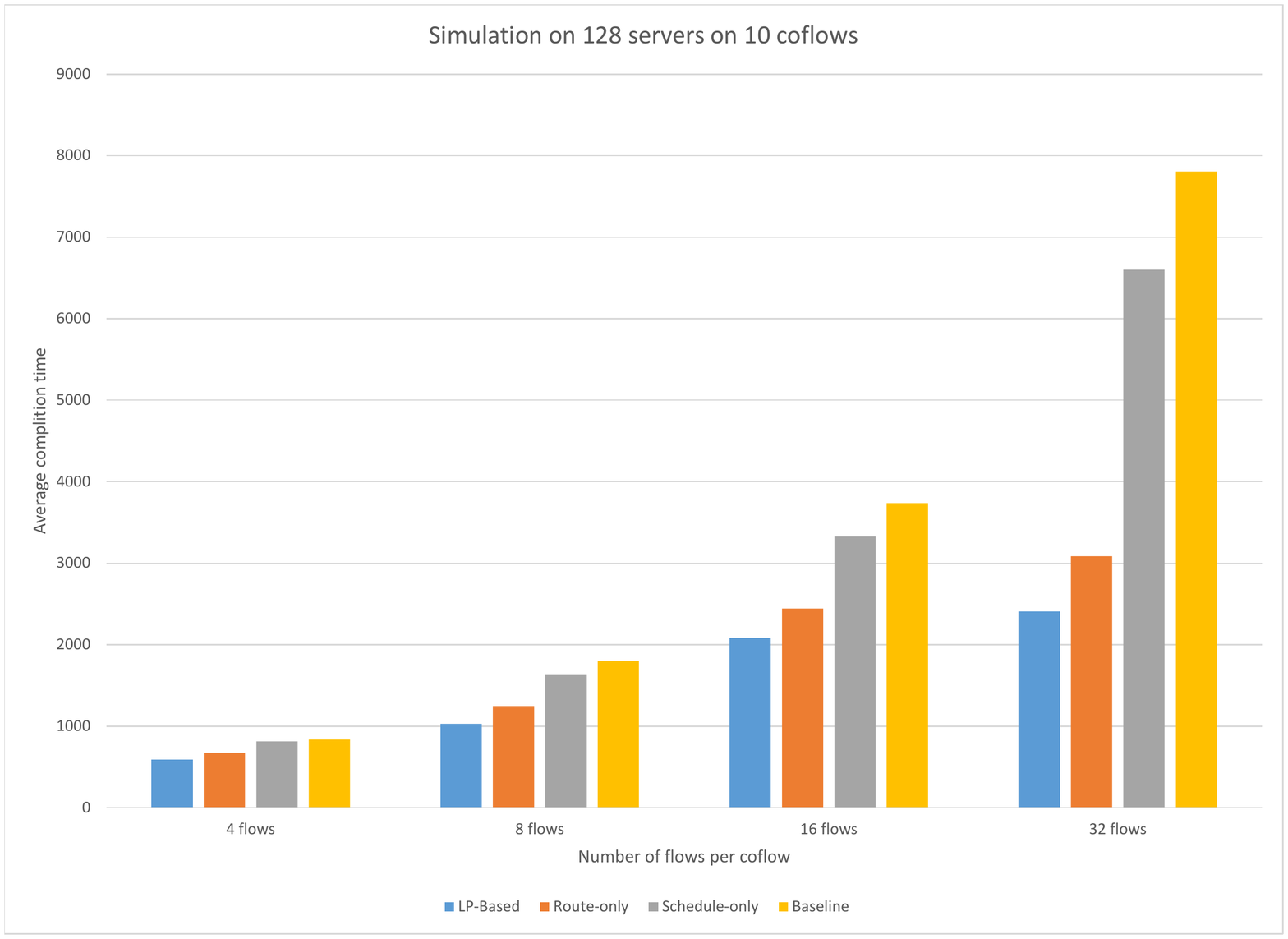}
\includegraphics[scale=0.3]{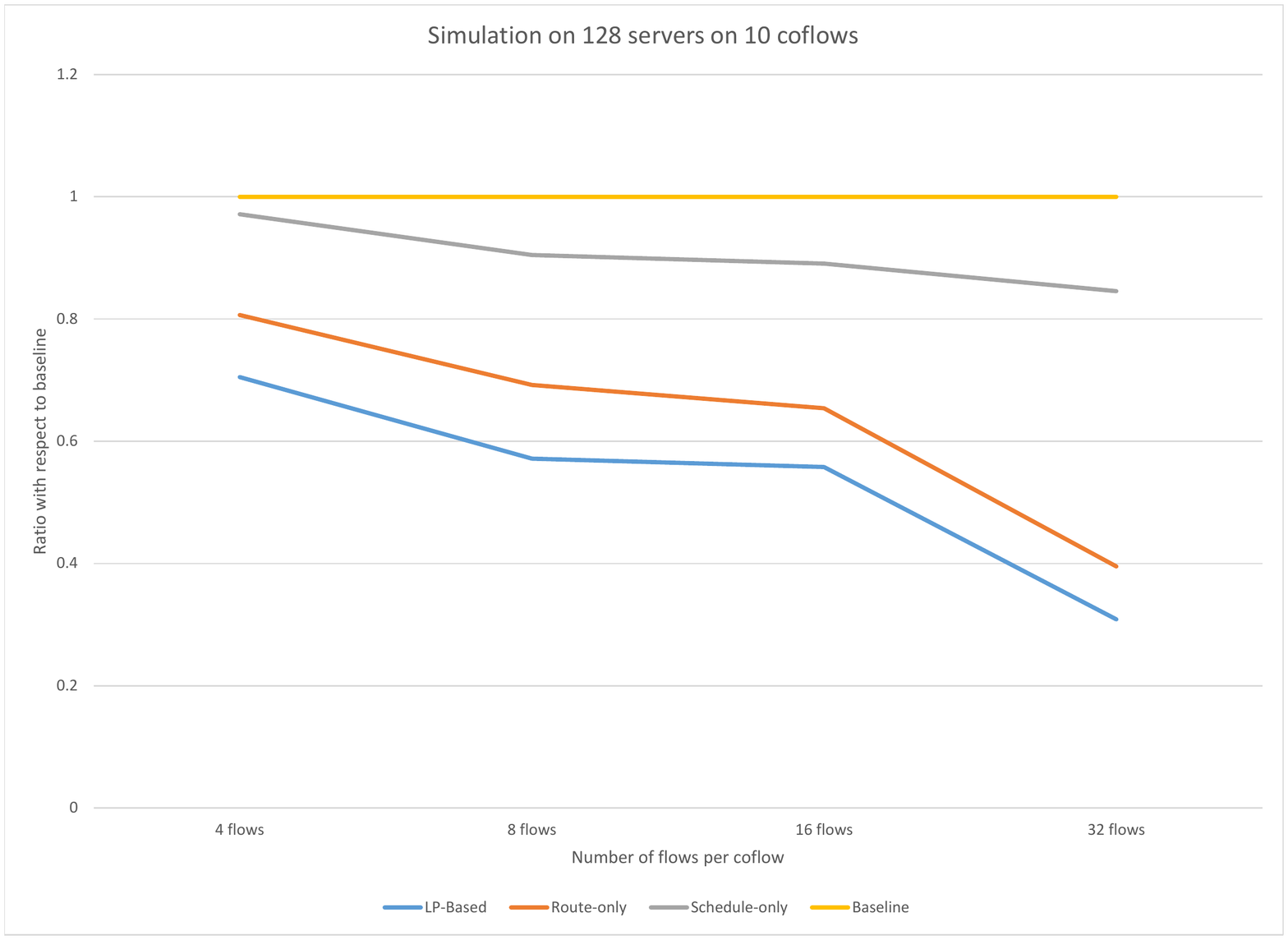}
\caption{ \label{fig:changeCoflowWidth} Changing the coflow width; the number of coflows is set to 10.
}
\end{center}
\end{figure}
\subsubsection*{Number of coflows} Using a fixed coflow width of 16, we vary the number of coflows from 10 to 25, in increments of 5.  The results are illustrated in Figure~\ref{fig:changeCoflowNum}. We find that again {\sc LP-Based} completes {\em every instance}\/ faster than the baseline as well as the Schedule-only and Route-only schedules.  On average, the improvement over the baseline, Schedule-only and Route-only is by factors of $\%110$, $\%72$, and $\%26$, respectively.  

\begin{figure}[h]
\begin{center}
\centering
\includegraphics[scale=0.3]{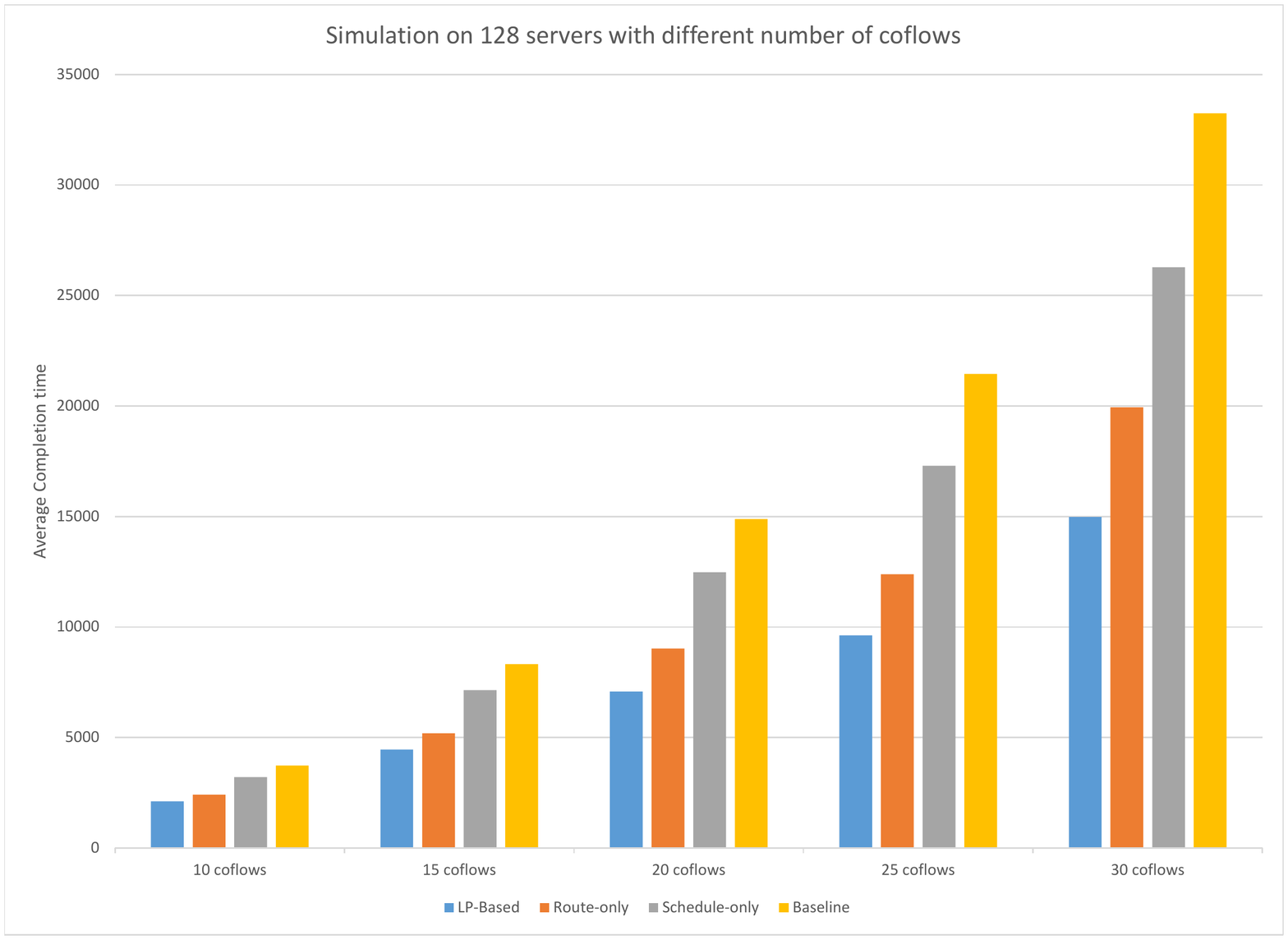}
\includegraphics[scale=0.3]{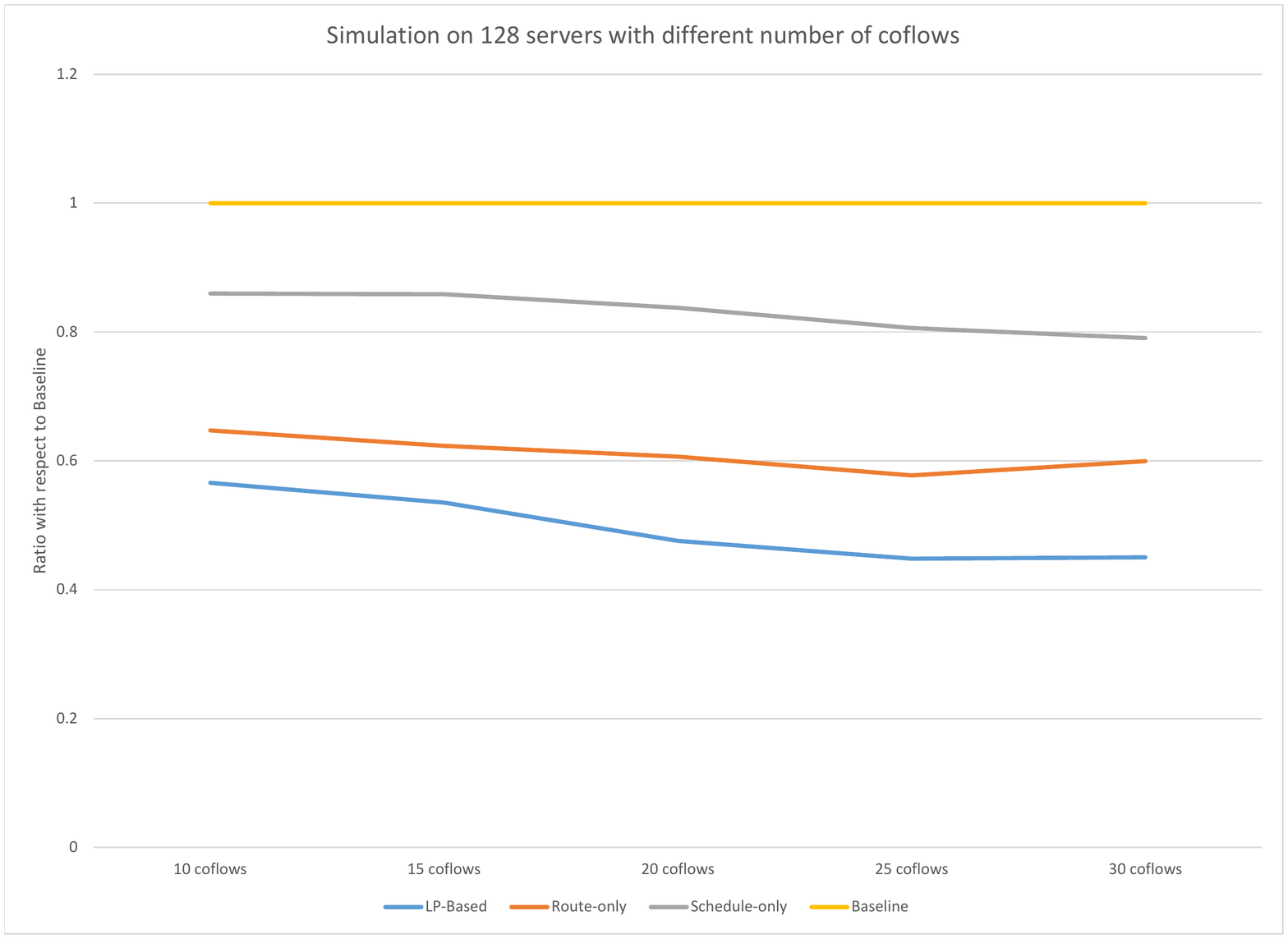}
\caption{ \label{fig:changeCoflowNum} Changing the number of coflows; coflow width is set to 16.
}
\end{center}
\end{figure}
We note that, as can be seen in the experiments, the worst-case approximation ratio of $O(\log |E| / \log\log |E|)$, does not happen in practice. The main reason is that in the fat-tree topology there are a few paths between pairs of servers; furthermore, uniform link bandwidth causes the path decomposition routine to return fewer flow paths. Indeed, in all of our experiments, the path decomposition routine returns one path per flow.
	
\section{Conclusion and Open Problems}\label{sec:open_problems}
We have presented algorithms for scheduling task-based coflows, packet-based coflows, and circuit-based coflows for general networks.  All of our algorithms achieve asymptotically-optimal total weighted completion time. An immediate line of further research is to obtain improved constant-factor bounds for general networks and improved algorithms for special network classes that are of particular interest to datacenter applications.  More broadly, we would like to pursue algorithms for other important objectives and coflow models.  Specifically, we are interested in designing approximation algorithms for the total weighted response time objective, where the response time is the difference between completion time and release time.  Another line of research is the design of schedules for coflows in optical networks.

\section{Acknowledgments}
We would like to thank the reviewers for their useful comments and pointing out the relevant work done in \cite{Correa2009}. 

\bibliographystyle{amsplain}
\bibliography{refs}

\end{document}